\documentclass[letterpaper, 10 pt, conference]{ieeeconf}
\IEEEoverridecommandlockouts
\usepackage{bbold,amsthm,amsfonts,amssymb}
\usepackage{mysty}
\usepackage[cmex10]{amsmath}
\usepackage{environ}
\usepackage{cite}
\usepackage{accents}
\usepackage{balance}
\usepackage{diagbox}
\usepackage{graphicx}
\usepackage{epsfig}
\usepackage{pbox}

\usepackage{ifthen}
\newboolean{showcomments}
\setboolean{showcomments}{true}
\usepackage[usenames,dvipsnames]{xcolor}
\usepackage{todonotes}

\definecolor{bleudefrance}{rgb}{0.19, 0.55, 0.91}
\definecolor{ao(english)}{rgb}{0.0, 0.5, 0.0}

\newcommand{\addcite}[0]{\ifthenelse{\boolean{showcomments}}
{\textcolor{purple}{(add cite(s)) }}{}}%

\newcommand{\enrique}[1]{  \ifthenelse{\boolean{showcomments}}
{\todo[inline,color=bleudefrance]{Enrique: #1}}{}}
\newcommand{\emmargin}[1]{\ifthenelse{\boolean{showcomments}}{\marginpar{\color{bleudefrance}\tiny EM: #1}}{}}

\newboolean{showedits}
\setboolean{showedits}{true}
\usepackage[xcolor=bleudefrance]{changes}
\definechangesauthor[color=bleudefrance]{EM}
\newcommand{\aem}[1]{
\ifthenelse{\boolean{showedits}}
{\added[id=EM]{#1}}
{#1}
}
\newcommand{\chem}[2]{
\ifthenelse{\boolean{showedits}}
{\replaced[id=EM]{#1}{#2}}
{#1}
}
\newcommand{\dem}[1]{
\ifthenelse{\boolean{showedits}}
{\deleted[id=EM]{#1}}
{}
}
\setboolean{showedits}{true}
\setboolean{showcomments}{true}
\usepackage{verbatim}

\newboolean{archive}
\setboolean{archive}{true}

\newtheorem{thm}{Theorem}
\newtheorem{lem}{Lemma}
\newtheorem{exmp}{Case}

\newenvironment{usethmcounterof}[1]{%
  \renewcommand\thethm{\ref{#1}}\thm}{\endthm\addtocounter{thm}{-1}}

\theoremstyle{remark}
\newtheorem*{rem}{Remark}

\begin{document}

\title{\LARGE Accurate Reduced-Order Models for Heterogeneous Coherent Generators}

%% To specify the authors when (number of affiliations <= 2)

\author{Hancheng Min, Fernando Paganini, and Enrique Mallada\thanks{H. Min and E. Mallada are with the Department of Electrical and Computer Engineering, Johns Hopkins University, Baltimore, MD 21218, USA {\tt\small \{hanchmin, mallada\}@jhu.edu}; F. Paganini is with Universidad ORT Uruguay, Montevideo, Uruguay {\tt\small paganini@ort.edu.uy}}
% \thanks{This work was supported by US DoE EERE award DE-EE0008006, NSF through grants CNS 1544771, EPCN 1711188, AMPS 1736448, and CAREER 1752362, Johns Hopkins University Discovery Award, and ANII-Uruguay, through grant FSE\_1\_2017\_1\_145060,}
% \IEEEauthorblockN{Hancheng Min, and Enrique Mallada}
% \IEEEauthorblockA{Johns Hopkins University, Baltimore, MD, U.S.A.\\
% Email: \{hanchmin, mallada\}@jhu.edu}
% \and
% \IEEEauthorblockN{Fernando Paganini}
% \IEEEauthorblockA{Universidad ORT Uruguay, Montevideo, Uruguay\\
% Email: paganini@ort.edu.uy}
% \thanks{This work was supported by US DoE EERE award DE-EE0008006, NSF through grants CNS 1544771, EPCN 1711188, AMPS 1736448, and CAREER 1752362, Johns Hopkins University Discovery Award, and ANII-Uruguay, through grant FSE\_1\_2017\_1\_145060,}
}

\maketitle

\begin{abstract}
We introduce a novel framework to approximate the aggregate frequency dynamics of coherent generators. By leveraging recent results on dynamics concentration of tightly connected networks, and frequency weighted balanced truncation, a hierarchy of reduced-order models is developed. This hierarchy provides increasing accuracy in the approximation of the aggregate system response, outperforming existing aggregation techniques. 
\end{abstract}

\section{Introduction}
Assessing performance in power grid frequency control requires models which are both accurate and tractable. In large-scale networks this goal has been sought for decades through \emph{aggregation} based on \emph{coherency} ~\cite{Chow2013}. Generally speaking, a group of generators is considered coherent if their bus frequencies exhibit a similar response when subject to power disturbances. These generators are often subsequently modeled by a single effective machine.   

Various methods for identifying coherent group of generators have been introduced in the past~\cite{podmore1978identification,Souza1992efficient,chow1982time,Winkelman1981,Nath1985}. The Linear Simulation Method~\cite{Podmore2013} groups generators whose maximum difference in time-domain response is within some tolerance. Similarly,  \cite{Souza1992efficient} develops a clustering algorithm based on the pairwise maximum difference in time-domain response. The Weak Coupling Method~\cite{Nath1985} quantifies strength of coupling between two areas to iteratively determine the boundaries of coherent generator groups. The Two Time Scale Method~\cite{chow1982time,Winkelman1981} computes the eigen-basis matrix associated with the electromechanical modes in the linearized network: generators with similar entries on the basis matrix with respect to low frequency oscillatory modes are considered coherent. 

Once generators are grouped by coherence, an effective machine model is typically proposed for each group. Previous work~\cite{Anderson1990,Germond1978,Guggilam2018,Apostolopoulou2016,ourari2006dynamic,Paganini2019tac} suggests that inertial and damping coefficients for the effective generator should chosen as the sum of the corresponding generator parameters. However, in the presence of turbine control dynamics, the proper choice of turbine time constants is unclear. Optimization-based approaches~\cite{Germond1978,Guggilam2018}  minimize an error function to choose the time constant of the effective generator. Other approaches use the average~\cite{Apostolopoulou2016}, or the weighted harmonic mean~\cite{ourari2006dynamic} of time constants of generators in the coherent group. \color{black} Accurate models of the coherent dynamics play an important role in applications to area dynamics modeling~\cite{ourari2006dynamic}, optimization of DER participation~\cite{Guggilam2018}, frequency shaping control~\cite{jiang2020grid}. Moreover, new modeling demands arise in modern-day networks where coherent groups may include grid-forming inverters\cite{schiffer2014conditions,tegling2015performance} in addition to classical synchronous generators. \color{black}

In this paper, we leverage new results \cite{min2019cdc} on characterizing coherence in tightly-connected networks to propose a general framework for aggregation of coherent generators. For $n$ coherent generators with transfer function $g_i(s),\ i=1,\cdots,n$, the aggregate coherent dynamics are accurately approximated by $\hat{g}(s)=\lp \sum_{i=1}^ng_i^{-1}(s)\rp^{-1}$. In particular, we show that $\hat{g}(s)$ is a natural characterization of the coherent dynamics in the sense that, as the algebraic connectivity of the network increases, the response of the coherent group is asymptotically $\hat{g}(s)$. Note, however, than in general due to heterogeneity in turbine control dynamics, the aggregate transfer function $\hat{g}(s)$ will be of an order which scales with the network size. We thus seek a low-order approximation.

In contrast with the conventional approach \cite{Germond1978,Guggilam2018,ourari2006dynamic} we will not restrict the choice of low order models to the simple selection of parameters of an effective generator. 
Rather, we will resort to \emph{frequency weighted balanced truncation} to develop a hierarchy of models of adjustable order and increasing accuracy. In particular, for an aggregation of $n$ second order generator models, we find that high accuracy can often be achieved by a reducing the $2n$-order system to 3rd order. We note however that the aggregation techniques introduced in this paper apply to any linear model of generators, including those of higher order than two.

We compare two alternatives: providing an aggregate model for a set of turbines, and subsequently closing the loop, versus performing the reduction directly on the closed loop  $\hat{g}(s)$. The first is motivated by retaining the interpretation whereby one or two equivalent generators represent the aggregate; still, we show how a similar interpretation may be available for the second, more accurate method. 

The rest of the paper is organized as follows. In Section \ref{sec_2_aggr_dym}, we provide the theoretical justification of the coherent dynamics $\hat{g}(s)$. In Section \ref{sec_3_red_ord_mdl}, we propose reduced-order models for $\hat{g}(s)$ by frequency weighted balanced truncation. We then show via numerical illustrations that the proposed models can achieve accurate approximation (Section \ref{sec_4_sim}). Lastly, we conclude this paper with more discussions on the implications of our current results. \color{black} A preliminary one-and-half page abstract of this work was presented in~\cite{min2019allerton}.\color{black}

\section{Aggregate Dynamics of Coherent Generators}\label{sec_2_aggr_dym}

% Consider a group of $n$ generators, indexed by $i=1,\cdots,n$, dynamically coupled through an AC network. The block diagram of the power system is shown in Fig.\ref{blk_power}.

% \begin{figure}[h!]
%     \centering
%     \includegraphics[width=5cm]{main/archive/pscc_blk_new.png}
%     \caption{Block Diagram of Power Networks}
%     \label{blk_power}
% \end{figure}

% For generator $i$, the transfer function from net power deviation at its generator axis to its angular frequency deviation $w_i$, relative to their equilibrium values, is given by $g_i(s)$. The net power deviation at generator $i$, includes disturbance $u_i$ reflecting variations in mechanical power or local load, minus the electrical power $p^e_i$ drawn from the network.

Consider a group of $n$ generators, indexed by $i=1,\cdots,n$ and dynamically coupled through an AC network. Assuming the network is in steady-state, Fig.\ref{blk_power} shows the block diagram of the linearized system around its operating point.

\begin{figure}[h!]
    \centering
    \includegraphics[width=4.5cm]{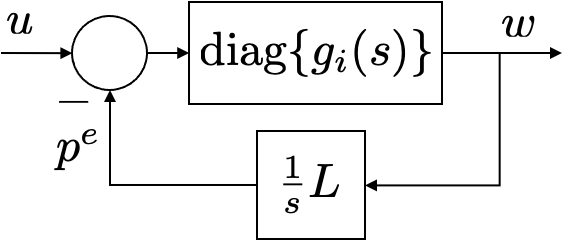}
    \caption{Block Diagram of Linearized Power Networks}
    \label{blk_power}
\end{figure}
\color{black} Due to the space constraints, we refer to~\cite{zhao2013power} for details on the linearization procedure. The signals $w=[w_1,\cdots,w_n]^T,u=[u_1,\cdots,u_n]^T,p^e=[p^e_1,\cdots,p^e_n]^T$ are in vector form. \color{black} For generator $i$, the transfer function from net power deviation $(u_i-p^e_i)$ at its generator axis to its angular frequency deviation $w_i$, relative to their equilibrium values, is given by $g_i(s)$. The net power deviation at generator $i$, includes disturbance $u_i$ reflecting variations in mechanical power or local load, minus the electrical power $p^e_i$ drawn from the network.

The network power fluctuations $p^e$ are given by a linearized (lossless) DC model of the power flow equation $p^e(s) = \frac{1}{s}Lw(s)$. Here $L$ is the Laplacian matrix of an undirected weighted graph, with its elements given by
$L_{ij} = \left.\frac{\partial}{\partial \theta_j}\sum_{k=1}^n|V_i||V_k|b_{ik}\sin(\theta_i-\theta_k)\right\vert_{\theta=\theta_0}\,,$
where $\theta_0$ are angles at steady state, $|V_i|$ is the voltage magnitude at bus $i$ and $b_{ij}$ is the line susceptance. Without loss of generality, we assume the steady state angular difference $\theta_{0i}-\theta_{0j}$ across each line is smaller than $\frac{\pi}{2}$. Moreover, because $L$ is a symmetric real Laplacian, its eigenvalues are given by $0=\lambda_1(L)\leq \lambda_2(L)\leq \cdots\leq \lambda_n(L)$.
The overall linearized frequency dynamics of the generators is given by
\begin{subequations}\label{eq:system}
\begin{align}
    w_i(s) &=\; g_i(s)(u_i(s)-p_i^e(s)),\quad i=1,\cdots,n\,,\label{eq_gen_dym}\\
    p^e(s)&=\; \frac{1}{s}Lw(s)\,.
\end{align}
\end{subequations}

\color{black}
Generally, a group of generator coupled as in Fig. \ref{blk_power} is considered \emph{coherent} if all generators have the same/similar frequency responses under disturbance $u$ of any shape. We are interested in characterizing the dynamic response of coherent generators, which we term here \emph{coherent dynamics}.  With this aim, we seek conditions on the network \eqref{eq:system} under which the entire set of generators behave coherently. The same approach can be used on subgroups of generators.

To motivate our results, we start with summing over all equations in \eqref{eq_gen_dym} to get
\be
    \sum_{i=1}^ng_i^{-1}(s) w_i(s) = \sum_{i=1}^nu_i(s)-\sum_{i=1}^np_i^e(s)=\sum_{i=1}^nu_i(s)\,.\label{eq_aggr_dym_m1}
\ee
Notice that the term $\sum_{i=1}^np_i^e(s)=\one^T\frac{L}{s}w(s)=0$ since $\one=[1,\cdots,1]^T$ is an left eigenvector of $\lambda_1(L)=0$. 

A pragmatical approach to obtain a model of coherent behavior is to simply \emph{impose} the equality  $w_i(s)=\hat w(s)$ between the frequency output. Solving from \eqref{eq_aggr_dym_m1} we obtain:
% at the output of the coherent generators~\cite{Germond1978,Anderson1990,Guggilam2018}, then \eqref{eq_aggr_dym_m1} can be written as
% \be
%     \lp\sum_{i=1}^ng_i^{-1}(s)\rp\hat{w}(s) = \sum_{i=1}^nu_i(s)\,. \label{eq_aggr_dym_m2}
% \ee
% From \eqref{eq_aggr_dym_m2}, the dynamics from the aggregate disturbance $\sum_{i}^nu_i(s)$ to the frequency response $\hat{w}(s)$ is given by 
\be
    \hat{w}(s)=\lp\sum_{i=1}^ng_i^{-1}(s)\rp^{-1}\sum_{i=1}^nu_i(s)=:\hat{g}(s)\sum_{i=1}^nu_i(s);\,
    \label{eq_aggr_dym}
\ee
the group of generators is aggregated into a single effective machine $\hat{g}(s)$, responding to the total disturbance. 
\color{black}
% To motivate our results, we follow the typical assumption, which is to impose an equal frequency response $w_i(s)=\hat w(s)$ at the output of the coherent generators~\cite{Germond1978,Anderson1990,Guggilam2018}, to derive a closed form expression for the coherent dynamics; the  theory that justifies the result of this derivation is then provided in Section \ref{ssec:coherence}.
% By assuming $w_i(s) = \hat{w}(s)$, it is possible to sum over all equations in \eqref{eq_gen_dym} to get
% \be
%     \lp\sum_{i=1}^ng_i^{-1}(s)\rp \hat{w}(s) = \sum_{i=1}^nu_i(s)-\sum_{i=1}^np_i^e(s)\,.\label{eq_aggr_dym_m1}
% \ee
% Notice that the last term $\sum_{i=1}^np_i^e(s)=\one^T\frac{L}{s}\one \hat{w}(s)=0$ since $\one=[1,\cdots,1]^T$ is an eigenvector of $\lambda_1(L)=0$. Then the aggregate model for the coherent group is given by
% \be
%     \hat{w}(s) = \lp\sum_{i=1}^ng_i^{-1}(s)\rp^{-1} \sum_{i=1}^nu_i(s)\,. \label{eq_aggr_dym_m2}
% \ee
% From \eqref{eq_aggr_dym_m2}, the coherent group of generators is aggregated into a single effective machine with its transfer function given by
% \be
%     \hat{g}(s) = \lp\sum_{i=1}^ng_i^{-1}(s)\rp^{-1}\,.
%     \label{eq_aggr_dym}
% \ee

%To properly substantiate equation \eqref{eq_aggr_dym} as our model for the coherent dynamcs, we will follow the procedure of ~\cite{min2019cdc}.
%for the characterization of coherence of tightly connected networks.

\subsection{Coherence in Tightly Connected Networks}\label{ssec:coherence}

To properly justify the use of  \eqref{eq_aggr_dym} as an accurate descriptor of the coherent dynamics, we state here a precise result.  Our analysis will highlight the role of the algebraic connectivity $\lambda_2(L)$ of the network as a direct indicator of how coherent a group of generators is.
% , which potentially can be used to identify coherent group in large power networks.

For the network shown in Fig.\ref{blk_power}, the transfer matrix from the disturbance $u$ to the frequency deviation $w$ is given by
\be
    T(s) = \lp I_n+\dg\{g_i(s)\}L/s\rp^{-1}\dg\{g_i(s)\}\,,
    \label{eq_T_explicit}
\ee
where $I_n$ is the $n\by n$ identity matrix. We establish that the transfer matrix $T(s)$ converges, as algebraic connectivity $\lambda_2(L)$ increases,  to one where all entries are given by $\hat g(s)$.

\color{black}

We make several assumptions: 1) $T(s)$ is stable; 2) $\hat{g}(s)$ in \eqref{eq_aggr_dym} is stable 3) all $g_i(s)$ are minimum phase systems. All generator network models discussed in this paper (Section \ref{subsec:aggr_dym},\ref{subsec:aggr_dym_mix}) satisfy these assumptions. In particular, the stability of $T(s)$ is guaranteed by passivity of the network~\cite{khalil2002nonlinear}. We state the following result.

\color{black}
\begin{thm}\label{thm_unifm_conv}
    Given the assumptions above, the following holds for any $\eta_0>0$:
    \ben
        \lim_{\lambda_2(L)\ra +\infty}\sup_{\eta\in[-\eta_0,\eta_0]}\lV T(j\eta)-\hat{g}(j\eta)\one\one^T\rV=0\,,
    \een
    where $j=\sqrt{-1}$ and $\one\in\mathbb{R}^n$ is the vector of all ones.
\end{thm}
\color{black}

The transfer matrix $\hat{g}(s)\one\one^T$ has the property that 
for an arbitrary vector disturbance $u(s)$, the response is
%\begin{equation*}%\label{eq:approx-response}
$w(s) = \hat{g}(s)\one\one^Tu(s) = \left(\hat g(s)\sum_{i=1}^nu_i(s)\right)\one$;  
%\end{equation*}
this says the vector of bus frequencies responds in unison, with all entries equal to the response $\hat{w}$ in \eqref{eq_aggr_dym}. 
Theorem \ref{thm_unifm_conv} states that in the limit of large connectivity, the true response  $T(s)u(s)$  is approximated by the one in \eqref{eq_aggr_dym} for the disturbances in the frequency band $[0,\eta_0]$.
\ifthenelse{\boolean{archive}}{The proof is shown in the appendix.}{Due to space constraints, we refer to \cite{min2019aggr} for the proof.}

% In other words, every bus frequency reacts to the aggregate disturbance $\sum_iu_i(s)$ based on the response $\hat g (s)$. As a result, %given that by Theorem \ref{thm_unifm_conv} 
% %$T(s)$ is close to $\hat{g}(s)\one\one^T$ on the frequency band $[0,\eta_0]$, 
% for any disturbance limited over band 

% $[0,\eta_0]$, the response of the. Therefore generator networks with large algebraic connectivity should be considered coherent and $\hat{g}(s)$ gives the coherent dynamics.

% \begin{figure}[!h]
%     \centering
%     \includegraphics[height=3.1cm]{main/archive/fig_iceland_approx_to_coi.png}
%     \caption{Left: Step response of Icelandic grid (generator responses and the CoI frequency response), and step response of coherent dynamics $\hat{g}(s)$; Right: approximation error to the CoI frequency(in terms of step response) of $\hat{g}(s)$ and two 2nd order aggregation models~\cite{Germond1978,Guggilam2018} (for comparison).}
%     \label{fig_approx_to_coi}
% \end{figure}

\begin{figure}[!h]
    \centering
    \includegraphics[height=3.3cm]{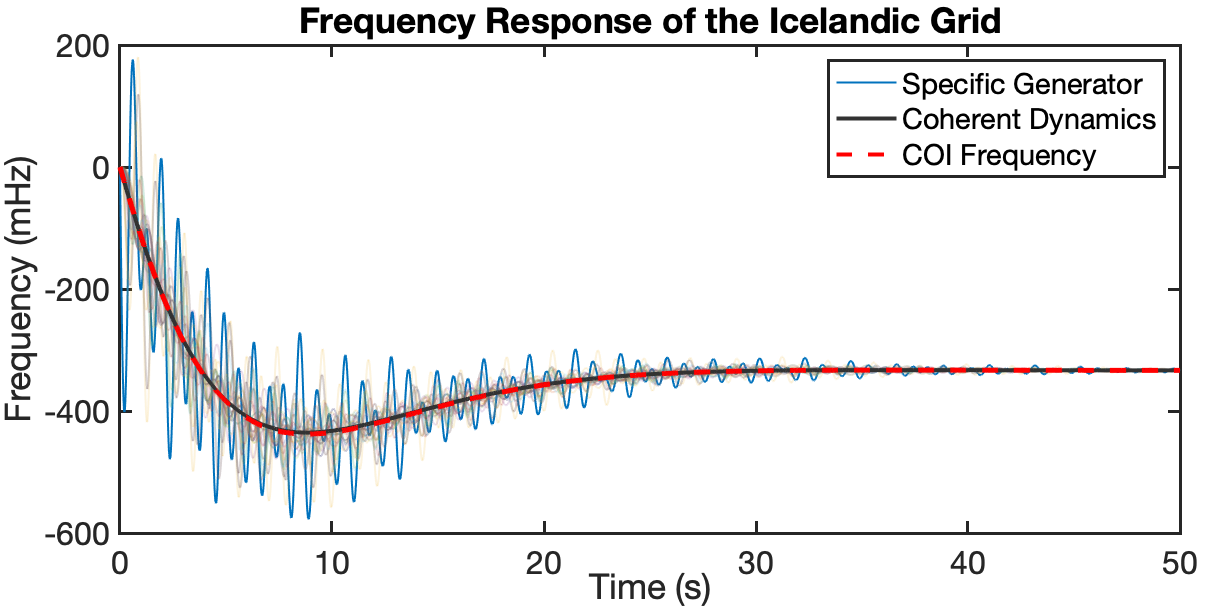}
    \caption{Step response of Icelandic grid (generator responses and the CoI frequency response), and step response of coherent dynamics $\hat{g}(s)$.  The  oscillatory response of a specific generator is highlighted in blue line.}
    \label{fig_approx_to_coi}
\end{figure}
The limit of high connectivity analyzed in the theorem is a good assumption for many cases of tightly connected networks, but one may wonder about the relevance of $\hat{g}(s)$ in a less extreme case. We explore this through a numerical  simulation on the Icelandic Power Grid~\cite{iceland}, of moderate connectivty. 
 As shown in Fig.\ref{fig_approx_to_coi}, the step response has incoherent oscillations from individual generators. Nevertheles,  if one looks at the Center of Inertia (CoI) frequency $w_{\mathrm{coi}}=(\sum_{i=1}^nm_i w_i)/(\sum_{i=1}^nm_i)$, a commonly used system-wide metric, we see it is very closely approximated by  the coherent dynamics $\hat{g}(s)$. Thus we will proceed with this model of aggregate response. For certain generator models, however, the  complexity of $\hat{g}(s)$ motivates the need for  approximations. 
 
%  approximate well the  which is generally considered a good characterization of aggregate frequency response of the network. Such observation encourages us to use $\hat{g}(s)$ to approximate the aggregate dynamics of power networks, but we will see next that $\hat{g}(s)$ could be complicated for some generator model.
% Theorem \ref{thm_unifm_conv} is presented as an asymptotic result; in practice the approximation  $T(s)\approx \hat{g}(s)\one\one^T$ will apply as long as  $\lambda_2(L)$ to be sufficiently large. 

% We list two cases where strong coupling exists among generators: 1) Multiple generators are connected to a single bus, where all generators can be approximated by a single aggregate generator; 2) A group of generators are physically close to each other so that the line susceptance is high between every pair of connected generators, or the connection among generators are dense. See 
%  ~\cite{Paganini2019tac} for numerical examples where adding lines to the network leads to a more coherent response. 

\color{black}

% Given , we should expect the system response of $T(s)$ to the disturbance $u(s)$ with its frequency components mainly from that low frequency band to be similar to those of $\hat{g}(s)\one\one^T$. And the responses of $\hat{g}(s)\one\one^T$ are exactly that every generator has the response of single effective machine $\hat{g}(s)$ under the aggregated disturbance $\one^Tu(s)$. 

\ifthenelse{\boolean{archive}}{As a side note, such coherence among generators is frequency-dependent. As we suggested above, the effective algebraic connectivity $\lv\frac{\lambda_2(L)}{s}\rv$ determines how close $T(s)$ is to $\hat{g}(s)\one\one^T$ at certain point. For any fixed $\lambda_2(L)$, there is a large enough cutoff frequency $\eta_c$ such that $\lv\frac{\lambda_2(L)}{j\eta}\rv$ is sufficiently small for any $\eta\geq \eta_c$, which is to say, for certain coherent group of generators, the responses of generators are not coherent at all under a disturbance with high frequency components over band $[j\eta_c,+\infty)$. }{}

\subsection{Aggregate Dynamics for Different Generator Models}\label{subsec:aggr_dym}

Having characterized how the \emph{coherent dynamics} given by $\hat{g}(s)$  represent the network's aggregate behavior, from now on we will use with no distinction the terms ``aggregate" and ``coherent" dynamics. Now we look into the explicit forms these dynamics take  for different generator models. 
\begin{exmp}\label{exmp_sw_mdl}
    Generators with 1st order model, of two types: 
    
    1) For synchronous generators\cite{Paganini2019tac}, 
    $
        g_i(s) =\frac{1}{m_is+d_i}\,,
    $
    where $m_i,d_i$ are the inertia and damping of generator $i$, respectively. The coherent dynamics are
    $
        \hat{g}(s) = \frac{1}{\hat{m}s+\hat{d}}\,,
    $
    where $\hat{m} = \sum_{i=1}^nm_i$ and $\hat{d} = \sum_{i=1}^nd_i$.
    
    2) For droop-controlled inverters\cite{schiffer2014conditions},
    $
        g_i(s)=\frac{k_{P,i}}{\tau_{P,i}s+1}\,,
    $
    where $k_{P,i}$ and $\tau_{P,i}$ are the droop coefficient and the filter time constant of the active power measurement, respectively. The coherent dynamics are
    $
        \hat{g}(s)=\frac{\hat{k}_P}{\hat{\tau}_Ps+1}\,,
    $
    where $\hat{k}_P=\lp\sum_{i=1}^nk_{P,i}^{-1}\rp^{-1},\ \hat{\tau}_P=\hat{k}_P\lp\sum_{i=1}^n\tau_{P,i}/k_{P,i}\rp$.
\end{exmp}
Notice that both dynamics are of the same form; by suitable reparameterization, we may use the ``swing" model $g_i(s) =\frac{1}{m_is+d_i}$ to model both types of generators. In this case no order reduction is needed: the aggregate model given in Case \ref{exmp_sw_mdl} is consistent with the conventional approach of choosing inertia $\hat m$ and damping $\hat d$ as the respective sums over all generators. Theorem \ref{thm_unifm_conv} explains why such a choice is indeed appropriate.

The aggregation is more complicated when considering generators with turbine droop control:
\begin{exmp}
    Synchronous generators given by the swing model with turbine droop\cite{Paganini2019tac}
    \be \label{eq_single_generator}
        g_i(s) = \frac{1}{m_is+d_i+\frac{r_i^{-1}}{\tau_is+1}}\,,
    \ee
    where $r_i^{-1}$ and  $\tau_i$ are the droop coefficient and turbine time constant of generator $i$, respectively. The coherent dynamics are given by 
    \be
        \hat{g}(s) = \frac{1}{\hat{m}s+\hat{d}+\sum_{i=1}^n\frac{r_i^{-1}}{\tau_is+1}}\,.\label{eq_aggr_dym_sw_tb}
    \ee
\end{exmp}
When all generators have the same turbine time constant $\tau_i=\hat \tau$, then $\hat{g}(s)$ in \eqref{eq_aggr_dym_sw_tb} reduces to the typical effective machine model of the form \eqref{eq_single_generator} with parameters $(\hat{m},\hat{d},\hat{r}^{-1},\hat{\tau}),$
where $\hat r^{-1}=\sum_{i=1}^nr_i^{-1}$, i.e., the aggregation model is still obtained by choosing parameters as the respective sums of their individual values. \color{black}However, if the $\tau_i$ are heterogeneous, then $\sum_{i=1}^n\frac{r_i^{-1}}{\tau_is+1}$ is generally high-order because the summands have distinct poles. As the result, the closed-loop dynamics $\hat{g}(s)$ is a high-order transfer function and cannot be accurately represented by a single generator model. \color{black} The aggregation of generators thus requires a low-order approximation of $\hat{g}(s)$.

\subsection{Aggregate Dynamics for Mixture of Generators}\label{subsec:aggr_dym_mix}
We have shown the aggregate dynamics for generators of three different types. When a mixture of these different types is present\footnote{Generally, when considering a mixture of synchronous generators and grid-forming inverters, our network model is valid only when synchronous generators make up a significant portion of the composition.}, we propose \eqref{eq_single_generator} to be a general representation of the three types; in particular, the first order models can be regarded as \eqref{eq_single_generator} with $r_i^{-1}=0$. Therefore, \eqref{eq_aggr_dym_sw_tb} provides a general representation of the aggregate dynamics resulting from a mixture of generators. Again, high-order coherent dynamics arise when heterogeneous turbines exist.

\section{Reduced Order Model for Coherent Generators with Heterogeneous Turbines}\label{sec_3_red_ord_mdl}

As shown in the previous section, the coherent dynamics $\hat{g}(s)$ are of high-order if the coherent group has generators with different turbine time constants. This suggests that substituting $\hat{g}(s)$ with an equivalent machine of the same order as each $g_i(s)$ may lead to substantial approximation error. In this section we propose instead a hierarchy of reduction models with increasing order, based on balanced realization theory~\cite{Zhou:1996:ROC:225507}, such that eventually an accurate reduction model is obtained as the order of the reduction increases. We further explore other avenues of improvement by applying the reduction methodology over the coherent dynamics itself, instead of the standard approach of applying a reduction only on the turbines~\cite{Germond1978,Guggilam2018,ourari2006dynamic}.

We use frequency weighted balanced truncation~\cite{SANGWOOKIM1995183} to approximate $\hat{g}(s)$. Frequency weighted balanced truncation identifies the most significant dynamics with respect to particular LTI frequency weight by computing the weighted Hankel singular values, which  decay fast in many cases, allowing us to accurately approximate high-order systems. Importantly, the reduction procedure favors approximation accuracy in certain frequency range specified by the weights. \ifthenelse{\boolean{archive}}{The detailed procedure of frequency weighted balanced truncation is shown in Appendix.\ref{app_freq_weight_bt}.}{Due to space constraints, we refer to \cite{min2019aggr} for the detailed procedure of frequency weighted balanced truncation.} Given a SISO proper transfer function $G(s)$, and a frequency weight $W(s)$ the $k$-th order weighted balanced truncation returns
% For an SISO stable and strictly proper system $G(s)$, the transfer function of $k$-th order reduced model by frequency weighted balanced truncation is in the form:
\begin{equation}\label{eq:transfer}
    \Tilde{G}_k(s)=\frac{b_{k-1}s^{k-1}+\cdots+b_1s+b_0}{a_{k}s^{k}+\cdots+a_1s+a_0}\,,
\end{equation}
which is guaranteed to be stable~\cite{SANGWOOKIM1995183}, and such that the weighted error $\sup_{\eta\in \mathbb{R}}|W(j\eta)(G(j\eta)-\tilde G_k(j\eta))|$ is upper bounded, with an upper bound decreasing to zero with the order $k$. For our purposes, $W(s)$ must have high gain in the low frequency range, so that the DC gains of the original and reduced dynamics are approximately matched, i.e., $G(0)\simeq \tilde{G}(0)$. Our proposed two model reduction approaches for high-order $\hat{g}(s)$ in \eqref{eq_aggr_dym_sw_tb} are both based on frequency weighted balanced truncation.

% \color{black}
% Another feature of interest of the reduction is its interpretability. Specifically, our reduced models can be rewritten in the form of a single generator as
% \be
%     \Tilde{g}_k(s) = \frac{1}{\tilde{m}s+\tilde{d}+\tilde{g}_t(s)}\,,
% \ee
% where $k$ is the order of $\Tilde{g}_k(s)$,  $\tilde{m},\tilde{d}$ can be viewed as the inertia and damping of the aggregate generator, and $\tilde{g}_t(s)$ may be interpreted as the turbine control dynamics. In the case of 2nd order model, $k=2$, we have $\tilde{g}_t(s)=\frac{\tilde{r}^{-1}}{\tilde{\tau}s+1}$, where  $\tilde{r}^{-1},\tilde{\tau}$ are the turbine parameters. Later we show that in the case $k=3$, $\tilde{g}_t(s)$ has an interesting interpretation as well.
% \color{black}
% \subsection{Frequency weighted balanced truncation on turbine dynamics}
\subsection{Model Reduction on Turbine Dynamics}\label{ssec:turbine-reduction}
Our first model is based on applying balanced truncation to the turbine aggregate. Essentially, $\hat{g}(s)$ in \eqref{eq_aggr_dym_sw_tb} is of high order because it has high-order turbine dynamics $\sum_{i=1}^n\frac{r_i^{-1}}{\tau_is+1}$; we seek to replace it with a reduced-order model.  This is akin to the existing literature ~\cite{Germond1978,Guggilam2018} which replaces an aggregate of turbines in parallel by a first order turbine model with parameters obtained by minimizing certain error functions.

We denote the aggregate turbine dynamics as
$
    \hat{g}_{t}(s):=\sum_{i=1}^n\frac{r_i^{-1}}{\tau_is+1}.
$
We also denote the $(k-1)$-th reduction model of $\hat{g}_{t}(s)$ by frequency-weighted balanced truncation as $\Tilde{g}_{t,k-1}(s)$. Then the $k$-th order reduction model of $\hat{g}(s)$ is given by
\be\label{eq_tf_bt_ndc}
    \Tilde{g}_{k}^{tb}(s)=\frac{1}{\hat{m}s+\hat{d}+\Tilde{g}_{t,k-1}(s)},
\ee
with, again, $\hat m =\sum_{i=1}^nm_i, \hat d=\sum_{i=1}^n d_i$. We highlight two special instances of relevance for our numerical illustration.

\subsubsection{2nd order reduction model}\label{ssec:turbine-k2}
When $k=2$, the reduced model $\Tilde{g}_{t,1}(s)$ can be interpreted as a first order turbine model  $$\Tilde{g}_{t,1}(s)=\frac{b_0}{a_1s+a_0}=\frac{b_0/a_0}{(a_1/a_0)s+1}:=\frac{\Tilde{r}^{-1}}{\Tilde{\tau}s+1}\,,$$ with parameters $(\Tilde{r}^{-1},\Tilde{\tau})$ chosen by the weighted balanced truncation method.
% It will give a single turbine model with parameter $(\Tilde{r}^{-1},\Tilde{\tau})$. 
Then the overall reduction model $\Tilde{g}_2^{tb}(s)$ is second order, which is a single generator model.

Unlike \cite{Germond1978,Guggilam2018}, there is a DC gain mismatch between $\Tilde{g}_{2}^{tb}(s)$ and the original $\hat{g}(s)$ since $\Tilde{r}^{-1}\neq \hat{r}^{-1}=\sum_{i=1}^nr_i^{-1}$. Later in the simulation section, by choosing a proper frequency weight $W(s)$, we effectively make the DC gain mismatch negligible. Unfortunately, as we will see in the numerical section, $k=2$ may not suffice to accurately approximate the coherent dynamics.

\subsubsection{3rd order reduction model}\label{ssec:turbine-k3}
To obtain a more accurate reduced-order model, one may consider $k=3$ as the next suitable option. In fact, we see in the later numerical simulation, a 2nd order turbine model $\Tilde{g}_{t,2}(s)$, i.e., $k=3$, is sufficient to give an almost exact approximation of $\hat{g}_{t}(s)$.

We can also interpret $\Tilde{g}_{t,2}(s)$, by means of partial fraction expansion, i.e.,
\ben
    \Tilde{g}_{t,2}(s)=\frac{b_1s+b_0}{a_2s^2+a_1s+a_0}=\frac{\Tilde{r}_1^{-1}}{\Tilde{\tau}_1s+1}+\frac{\Tilde{r}_2^{-1}}{\Tilde{\tau}_2s+1},
\een
assuming the poles are real. Then the reduced dynamics $\Tilde{g}_{t,2}(s)$ can be viewed as two first order turbines in parallel with parameters $(\Tilde{r}_1^{-1},\Tilde{\tau}_1)$ and $(\Tilde{r}_2^{-1},\Tilde{\tau}_2)$. In Section \ref{subsec:eff_red_order}, we show such interpretation is valid for our numerical example.

% \subsection{Frequency weighted balanced truncation on closed-loop dynamics}
\subsection{Model Reduction on Closed-loop Coherent Dynamics}\label{ssec:red_cl}
Our second proposal is: instead of reducing the turbine dynamics \eqref{eq_tf_bt_ndc}, to apply weighted balanced truncation directly on $\hat g(s)$. Thus, we denote $\Tilde{g}_k^{cl}(s)$ as the $k$-th order reduction model, via frequency weighted balanced truncation, of the coherent dynamics $\hat{g}(s)$. Again, DC gain mismatch can be made negligible by properly choosing $W(s)$.

As compared to  Section \ref{ssec:turbine-reduction}, the reduced model might not be easy to interpret in practice. Nevertheless, the procedure described below often leads to such an interpretation. 

\subsubsection{2nd order reduction model} When $k=2$, we wish to interpret $\Tilde{g}_2^{cl}(s)$ in terms of a single generator with a first order turbine of the form in  \eqref{eq_single_generator},  with parameters $(\Tilde{m},\Tilde{d},\Tilde{r}^{-1},\Tilde{\tau})$.
Given  $$\Tilde{g}_2^{cl}(s)=\frac{b_1s+b_0}{a_2s^2+a_1s+a_0}:=\frac{N(s)}{D(s)}\,,$$ obtained via the proposed method, we write the polynomial division $D(s) = Q(s)N(s)+R$, where $Q(s),R$ are quotient and remainder, respectively. This leads to the expression 
\ben
    \Tilde{g}_2^{cl}(s)=\frac{N(s)}{Q(s)N(s)+R}=\frac{1}{Q(s)+\frac{R}{N(s)}}\,.
\een
Here the first order polynomial $Q(s)$ can be matched to $\Tilde{m}s+\Tilde{d}$, and $\frac{R}{N(s)}$  to $\frac{\Tilde{r}^{-1}}{\Tilde{\tau}s+1}$. Provided the obtained constants $(\Tilde{m},\Tilde{d},\Tilde{r}^{-1},\Tilde{\tau})$ are positive, the interpretation follows.

\subsubsection{3rd order reduction model} 
Similarly, when $k=3$, the reduced model is $\Tilde{g}_3^{cl}(s)=\frac{N(s)}{D(s)}$, with $N(s)$ of 2nd order and $D(s)$ of 3rd order. The polynomial division $D(s) = Q(s)N(s)+R(s)$, still gives a first order quotient $Q(s)$, which  is interpreted as $\Tilde{m}s+\Tilde{d}$; the second order transfer function 
$\frac{R(s)}{N(s)}$ can be expressed, by partial fraction expansion, as two first order turbines in parallel, provided the obtained constants remain positive. We explore this in the examples studied below.

\section{Numerical Simulations}\label{sec_4_sim}

We now evaluate the reduction methodologies proposed in the previous section, and compare their performance with the solutions proposed in~\cite{Germond1978,Guggilam2018}. 
In our comparison, we consider 5 generators forming a coherent group\footnote{More specifically, we assume sufficiently strong network coupling among these generators such that the frequency responses are coherent. The numerical simulation will only illustrate the approximation accuracy with respect to the coherent response rather than individual ones.}. All parameters are expressed in a common base of 100 MVA.

\emph{The test case}: 5 generators, $\hat{m}=0.0683 (\text{s}^2/\text{rad})$, $\hat{d}= 0.0107$. The turbine and droop parameters of each generator are listed  in Table \ref{tb_droop_param_case1}. In all comparisons, a step change of $-0.1$ p.u. is used.

\begin{table}[!h]
\centering
\caption{Droop control parameters of generators in test case}
\begin{tabular}{l|lllll}
\hline
\diagbox[width=2.5cm,height=0.6cm]{Parameter}{Index}& 1      & 2      & 3      & 4      & 5      \\ \hline
droop $r_i^{-1}$ (p.u.) & 0.0218 & 0.0256 & 0.0236 & 0.0255 & 0.0192 \\ \hline
time constant $\tau_i$ (s)      & 9.08   & 5.26   & 2.29   & 7.97   & 3.24   \\ \hline
\end{tabular}

\label{tb_droop_param_case1}
\end{table}

\begin{rem}
    In the test case, we only aggregate 5 generators and report all  parameters explicitly in order to give more insights on how the distribution of time constant $\tau_i$ affects our approximations. It is worth noting that similar behavior is observed when reducing coherent groups with a much larger number of generators. In particular,  the accuracy found below with 3rd order reduced models is also observed in these higher order problems. 
\end{rem}

\ifthenelse{\boolean{archive}}
{
% if archive
\subsection{DC Gain Mismatch Cancellation}
As mentioned in the previous section, one of the drawbacks of the balanced truncation method is that it does not match the DC gain of the original system, which leads to an error on the steady-state frequency. We illustrate this issue in  Fig. \ref{fig_dc_vs_ndc}, where we compare the step response of two 2nd order reduction models $\Tilde{g}_2^{tb}(s)$ using frequency weighted balanced truncation on the turbines, with different weights: 1) unweighted: $W_1(s)=1$; 2) weighted: $W_2(s)=\frac{s+3\cdot 10^{-2}}{s+10^{-4}}$. 

Fig. \ref{fig_dc_vs_ndc} compares step responses and Bode plots for the original coherent dynamics $\hat{g}(s)$ (solid gray) with those of reduced models (dotted and dashed lines).

\begin{figure}[ht]
    \centering
    \includegraphics[width=8cm]{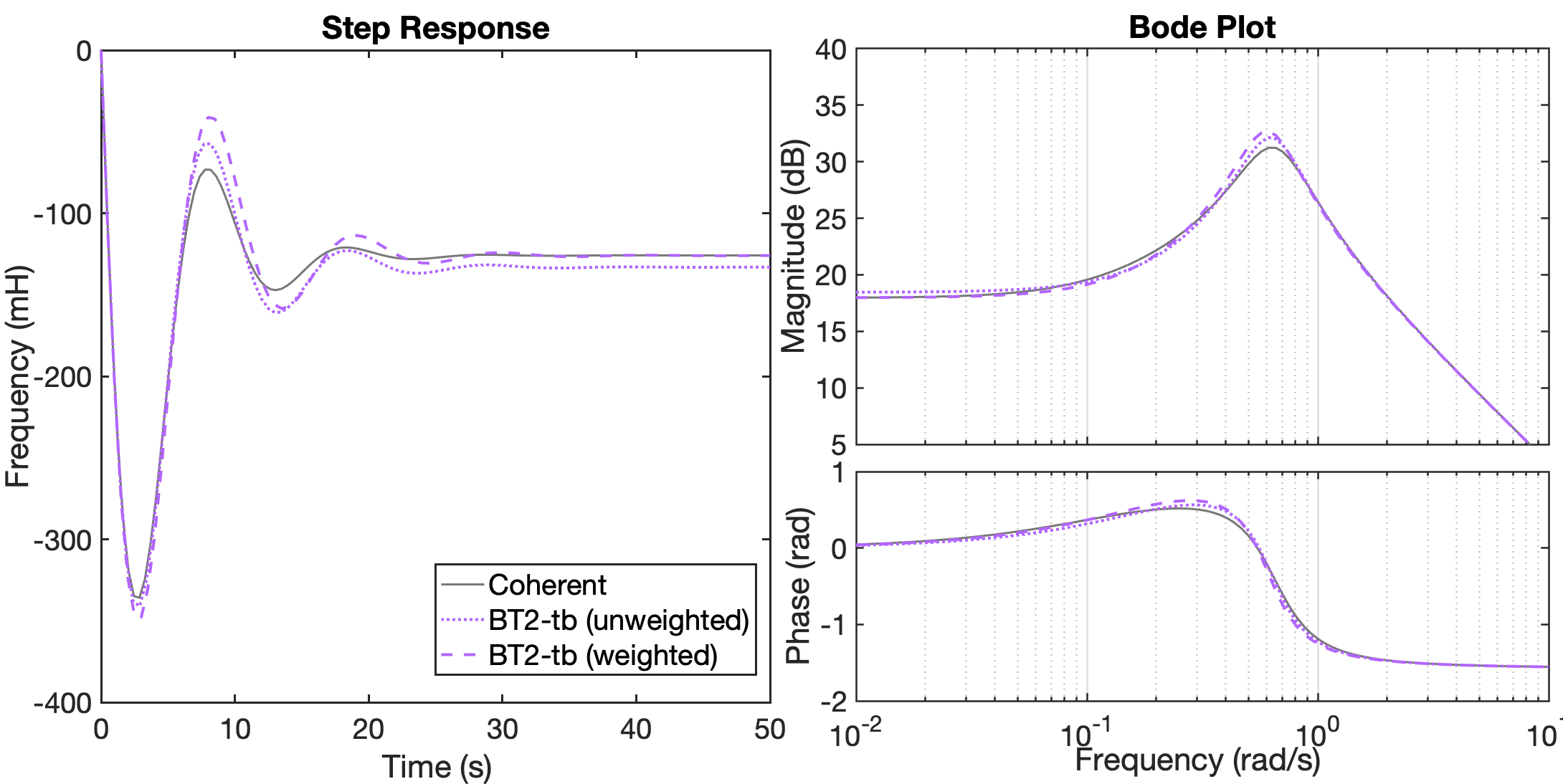}
    \caption{Second order models by balanced truncation on turbine dynamics with frequency weights $W_1(s)=1$ (unweighted) and  $W_2(s)=\frac{s+3\cdot 10^{-2}}{s+10^{-4}}$ (weighted).
    Step response (left) and Bode plot (right).}
    \label{fig_dc_vs_ndc}
\end{figure}

The DC gain mismatch is reflected in the steady state step response; we see that it is significantly reduced by  frequency weighted balanced truncation. However, it gives worse approximation to $\hat{g}(s)$ in the transient phase than the unweighted truncation. The Bode plot also reflects such a trade-off: the unweighted model has lower approximation error around the peak gain ($0.1-1$ rad/s) of $\hat{g}(s)$, at the cost of inaccuracies in the low frequency range ($<0.1$rad/s). The  weighted model exhibits exactly the opposite behavior, as the weight $W_2(s)=\frac{s+3\cdot 10^{-2}}{s+10^{-4}}$ puts more emphasis on low frequency ranges. 

As we will show in Section \ref{ssec:comparison}, neither can optimization-based approaches get rid of this trade-off. This suggests that a second order model is not sufficient to fully recover our coherent dynamics $\hat{g}(s)$. The main reason is that the time constants $\tau_i$ have wide spread: from $ {\raise.17ex\hbox{$\scriptstyle\sim$}}2$s to ${\raise.17ex\hbox{$\scriptstyle\sim$}}9$s. As the result, it is difficult to find a proper time constant $\Tilde{\tau}$ to account for both fast and slow turbines. The way to resolve it is approximating $\hat{g}(s)$ by higher-order reduced models.
}{
% else
As mentioned in the previous section, one of the drawbacks of the balanced truncation method is the DC gain mismatch, which leads to a steady-state error. In our simulation, the DC gain mismatch is effectively cancelled by picking proper frequency weights for different reduced models. Due to space constraints, we refer to \cite{min2019aggr} for the comparison between reduced models with and without frequency weights.
}

\subsection{Effect of Reduction Order $k$ in Accuracy}\label{subsec:eff_red_order}
We now evaluate the effect of the order of the reduction in the accuracy.
That is, we compare 2nd and 3rd order balanced truncation on the turbine dynamics, $\Tilde{g}_2^{tb}(s)$ (BT2-tb), $\Tilde{g}_3^{tb}(s)$ (BT3-tb), as well as balanced truncation on the closed-loop coherent dynamics $\Tilde{g}_2^{cl}(s)$ (BT2-cl), $\Tilde{g}_3^{cl}(s)$ (BT3-cl). The  frequency weights are given by $W_{tb}(s)=\frac{s+3\cdot 10^{-2}}{s+10^{-4}}$  and $W_{cl}(s)=\frac{s+8\cdot 10^{-2}}{s+10^{-4}}$, respectively.
% We compare all proposed reduction models by balanced truncation: 2nd and 3rd order balanced truncation on turbine dynamics $\Tilde{g}_2^{tb}(s)$ (BT2-tb), $\Tilde{g}_3^{tb}(s)$ (BT3-tb) with frequency weight $W_{tb}(s)=\frac{s+3\cdot 10^{-2}}{s+10^{-4}}$; 2nd and 3rd order balanced truncation on closed-loop dynamics, $\hat g(s)$, given by $\Tilde{g}_2^{cl}(s)$ (BT2-cl), $\Tilde{g}_3^{cl}(s)$ (BT3-cl) with frequency weight $W_{cl}(s)=\frac{s+8\cdot 10^{-2}}{s+10^{-4}}$. 
The step response and step response error with respect to $\hat{g}(s)$ are shown in Fig. \ref{fig_comparison_all_bt}.

\begin{figure}[ht]
    \centering
    \includegraphics[width=8.0cm,height=3.2cm]{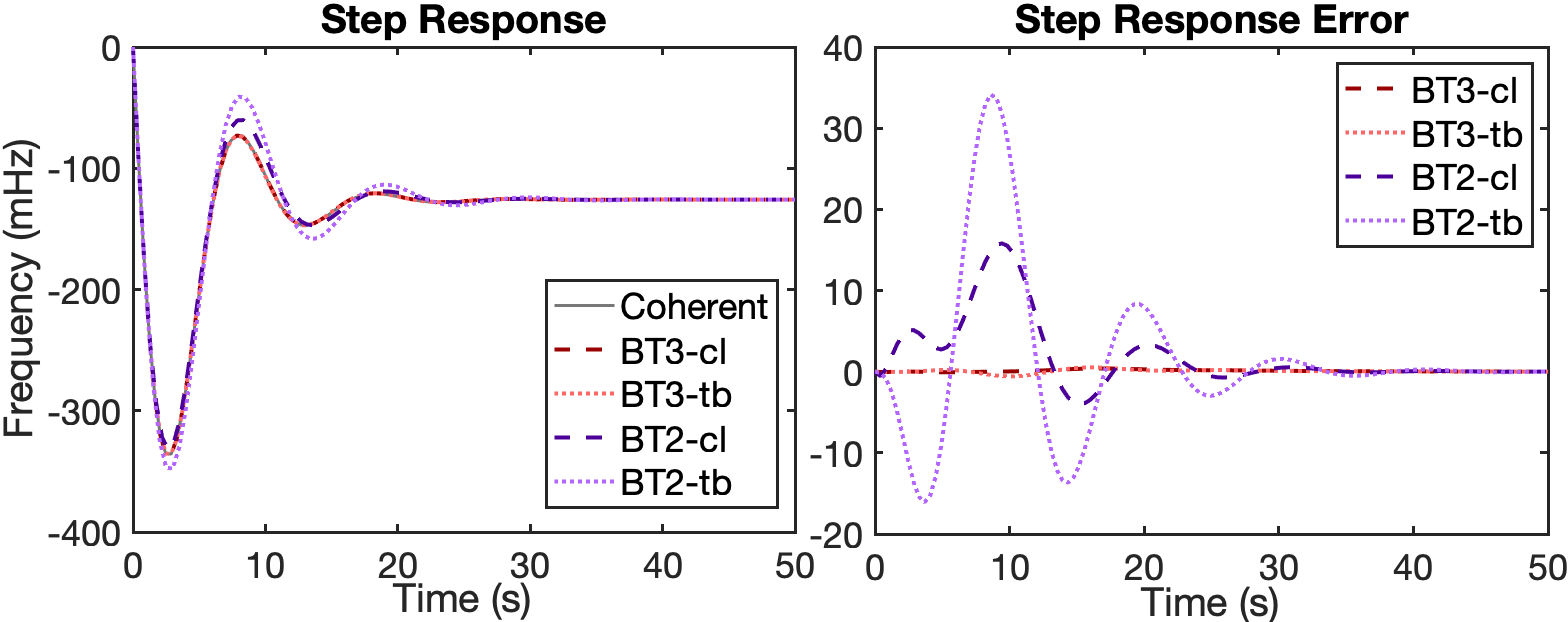}
    \caption{Comparison of all reduced-order models by balanced truncation}
    \label{fig_comparison_all_bt}
\end{figure}
Compared to 2nd order models, 3rd order reduced models give a very accurate approximation of $\hat{g}(s)$. While it is not surprising that approximation models with higher order ($k=3$) outperform models with lower order ($k=2)$, it is not trivial that a 3rd order model would provide this level of accuracy for an intrinsically high order system. 

Moreover, when we examine the transfer function given by $\Tilde{g}_3^{tb}(s)$ (from input $u$ in p.u. to output $w$ in rad/s), we find an interesting interpretation.
That is, the turbine model for $\Tilde{g}_3^{tb}(s)$ is given by 
\ben
    \Tilde{g}_{t,2}(s)=\frac{0.0266s+0.0057}{s^2+0.5046s+0.0489}=\frac{0.0473}{2.68s+1}+\frac{0.0684}{7.64s+1},
\een
where the latter is obtained by partial fraction expansion and can be viewed as two turbines (one fast turbine and one slow turbine) in parallel, and the choices of droop coefficients for these two turbines reflects the aggregate droop coefficients of fast turbines (generators 3 and 5) and slow turbines (generators 1,2, and 4), respectively, in $\hat{g}(s)$.

\subsection{Reduction on Turbines vs. Closed-loop Dynamics}

Another observation from Fig. \ref{fig_comparison_all_bt} is that reduction on the closed-loop is more accurate than reduction on the turbine. To get a more straightforward comparison, we list in Table \ref{tb_approx_err_comp} the approximation errors of all 4 models in Fig \ref{fig_comparison_all_bt} using the following metrics: 1) $\mathcal{L}_2$-norm of step response error\footnote{For reduced-order models obtained via frequency weighted balanced truncation, there exists an extremely small but non-zero DC gain mismatch that makes the $\mathcal{L}_2$-norm unbounded. We resolve this issue by simply scaling our reduced-order models to have exactly the same DC gain as $\hat{g}(s)$. } $e(t)$ (in $\mathrm{rad}/\mathrm{s}^{1/2}$): $(\int_{0}^{+\infty}|e(t)|^2dt)^{1/2}$; 2) $\mathcal{L}_\infty$-norm of $e(t)$ (in $\mathrm{rad}/\mathrm{s}$): $\max_{t\geq 0}|e(t)|$; 3) $\mathcal{H}_\infty$-norm difference between reduced and original models (from input $u$ in p.u. to output $w$ in rad/s).

% \begin{rem}
%     For reduced-order models obtained via frequency weighted balanced truncation, there exists extremely small but non-zero DC gain mismatch that makes $\mathcal{L}_2$-norm numerically unstable. We resolve this by simply scaling our reduced-order models to have exactly same DC gain as $\hat{g}(s)$. 
% \end{rem}
% \input{main/graphs/tables/tb_approx_err_1.tex}

\begin{table}[!h]
\centering
\caption{Approximation errors of reduced order models}
\begin{tabular}{l|lll}
\hline
\diagbox[width=2.4cm]{Model}{Metric}
  & \pbox{20cm}{$\mathcal{L}_2$ diff. \\ ($\mathrm{rad}/\mathrm{s}^{1/2}$)}& \pbox{20cm}{$\mathcal{L}_\infty$ diff. \\ ($\mathrm{rad}/\mathrm{s}$)}& $\mathcal{H}_\infty$ diff.       \\ \hline
Guggilam\cite{Guggilam2018}         & 7.2956               & 3.8287                     & 10.2748                           \\ \hline
Germond\cite{Germond1978}          & 3.9594                & 1.9974                     & 5.1431                           \\ \hline
BT2-tb        & 4.3737              & 2.1454                     & 7.5879                           \\ \hline
BT2-cl          & 2.0376                & 0.9934                     & 2.0381                           \\ \hline
BT3-tb   & 0.0967                & 0.0361                     & 0.1315                           \\ \hline
BT3-cl   & \textbf{0.0704}                & \textbf{0.0249}                    & \textbf{0.0317}                           \\ \hline
\end{tabular}
\label{tb_approx_err_comp}
\end{table}

We observe from Table \ref{tb_approx_err_comp} that for a given the reduction order, balanced truncation on the closed-loop dynamics ($\Tilde{g}_2^{cl}(s)$, $\Tilde{g}_3^{cl}(s)$) has smaller approximation error than balanced truncation on turbine dynamics ($\Tilde{g}_2^{tb}(s)$, $\Tilde{g}_3^{tb}(s)$) \emph{across all metrics}. Such observation seems to be true in general. For instance, Fig.~\ref{fig_ncl_vs_cl} shows a similar trend by plotting the same configuration (metrics and models) of Table \ref{tb_approx_err_comp} for different values of of the aggregate inertia $\hat{m}$, while keeping all other parameters the same. 
% Again, it can be seen that reduction on closed-loop dynamics always outperforms reduction on turbine dynamics.
\begin{figure}[ht]
    \centering
    \includegraphics[width=8cm,height=3.2cm]{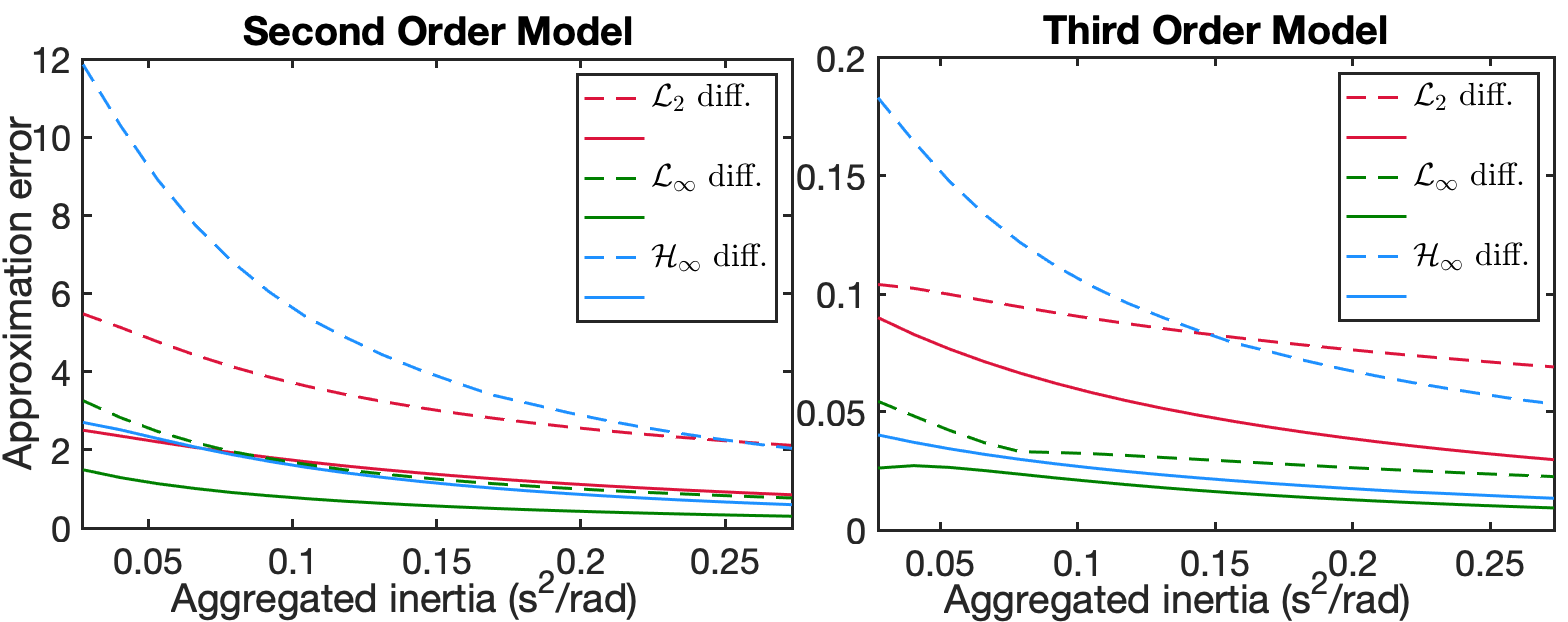}
    \caption{Approximation errors of second order models (left) and third order models (right) by balanced truncation in different metrics. Approximation errors of reduced-order models $\Tilde{g}_2^{tb}(s)$, $\Tilde{g}_3^{tb}(s)$ are shown in dashed lines; Approximation errors of reduced-order models $\Tilde{g}_2^{cl}(s)$, $\Tilde{g}_3^{cl}(s)$ are shown in solid lines. The approximation errors are in their respective unit.}
    \label{fig_ncl_vs_cl}
\end{figure}

It can be seen from Fig.~\ref{fig_ncl_vs_cl} that reduction on closed-loop dynamics improves the approximation in every metric, uniformly, for a wide range of aggregate inertia $\hat{m}$ values. 
The main reason is that, when applying reduction on closed-loop dynamics, the algorithm has the flexibility to choose the corresponding values of  inertia and damping to be  different from the aggregate ones in order to better approximate the response. More precisely,  from the reduced model we obtain
 \begin{align*}
     \Tilde{g}_{2}^{cl}(s)&=\;\frac{4.9733s+1}{(0.06715s+0.01464)(4.9733s+1)+0.1118}\,,
\end{align*}
from which we can get the equivalent swing and turbine model as
\begin{align*}
    \text{swing model:} &\; \frac{1}{0.06715s+0.01464},\,
    \text{turbine:} &\; \frac{0.1118}{4.9733s+1}.
\end{align*}
The equivalent inertia and damping are $\Tilde{m}=0.06715$ and $\Tilde{d}=0.01464$, which are different from the aggregate values $\hat{m},\hat{d}$. Therefore, when compared to reduction on turbine dynamics, reduction on closed-loop dynamics is essentially less constrained on the parameter space, thus achieving smaller approximation errors.

\subsection{Comparison with Existing Methods}\label{ssec:comparison}

Lastly, we compare reduced-order models via balanced truncation on the closed-loop dynamics, $\Tilde{g}_2^{cl}(s)$, $\Tilde{g}_3^{cl}(s)$, with the solutions proposed in~\cite{Germond1978,Guggilam2018}. The step responses and the approximation errors are shown in Fig.~\ref{fig_comparison_1} and Table.~\ref{tb_approx_err_comp}. 
\begin{figure}[ht]
    \centering
    \includegraphics[width=8.0cm,height=3.2cm]{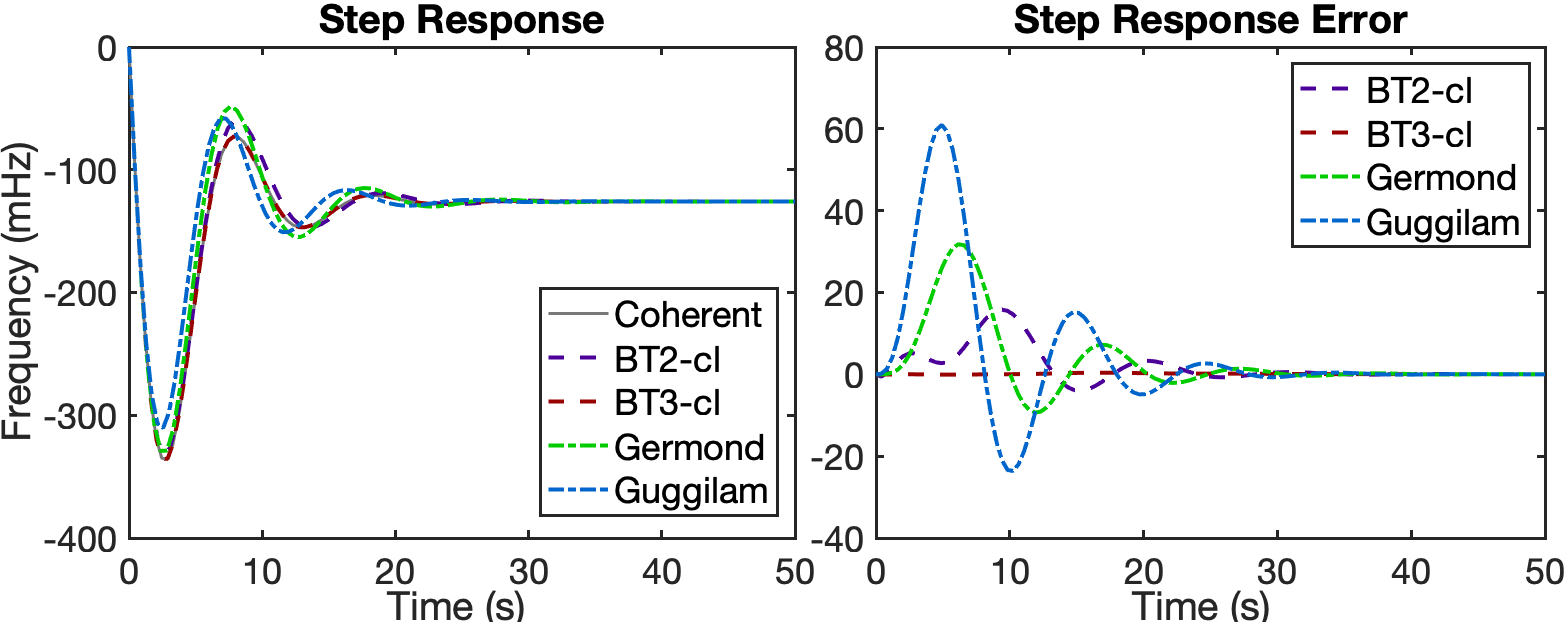}
    \caption{Comparison with existing reduced-order models}
    \label{fig_comparison_1}
\end{figure}

In the comparison, $\Tilde{g}_3^{cl}(s)$ outperforms all other reduced-order models and it is the most \textbf{accurate reduced-order model} of $\hat{g}(s)$. It is also worth noting that $\Tilde{g}_2^{cl}(s)$ has the least approximation error among all 2nd order models. In general, such results suggest us that to improve the accuracy in reduced-order models of coherent dynamics of generators $\hat{g}(s)$, we should consider: 1) increasing the complexity (order) of the reduction model; 2) reduction on closed-loop dynamics instead of on turbine dynamics. 

\section{Conclusion and Future Work}
This paper concerns tractable models for frequency dynamics in the power grid, starting with  the characterization  $\hat{g}(s)=\lp\sum_{i=1}^ng_i^{-1}(s)\rp^{-1}$ for the coherent response, which is shown to be asymptotically accurate as the coupling between generators (characterized via $\lambda_2(L)$) increases. Our characterization justifies existing aggregation approaches and also explains the difficulties of aggregating generators with heterogeneous turbine time constants. We leverage model reduction tools from control theory to find accurate reduced-order approximations to $\hat{g}(s)$. For $\{g_i(s)\}_{i=1}^n$ given by the 2nd order generator models, the numerical study shows that 3rd order models based on frequency weighted balanced truncation on closed-loop dynamics are sufficient to accurately represent $\hat{g}(s)$.

\color{black}
There are many possible directions of further inquiry. First, for situations of weaker coherency, we have seen that in Fig.\ref{fig_approx_to_coi}, $\hat{g}(s)$ well approximates the response of the CoI frequency, and in fact the approximation is exact if generator transfer functions $\{g_i(s)\}_{i=1}^n$ are proportional to each other~\cite{Paganini2019tac}. An interesting question is to bound the approximation error when proportionality fails. 

% it is shown in 
% ~\cite{Paganini2019tac} that if generator transfer functions $\{g_i(s)\}_{i=1}^n$ are proportional to each other, 
% $\hat{g}(s)$ describes exactly the response of the Center of Inertia (COI) frequency $\bar{w}=\lp\sum_{i=1}^nm_iw_i\rp/\lp\sum_{i=1}^nm_i\rp$. An interesting question is to provide approximatations to the COI response, when proportionality fails. 

A second topic of future research is experimentation with higher-order $g_i(s)$ \cite{Germond1978}. These arise due to more detailed models of turbine dynamics, or to the presence of advanced  droop controllers~\cite{jiang2017performance,jiang2020grid}. Classical aggregation strategies are complicated in this setting, but our model reduction program is in principle applicable and can help validate the need for such level of modeling detail.

% , which makes it challenging to find a reduced order model that is interpretable as a single aggregate generator. Additionally, the presence of advanced also contributes to a complicated generator models.
\color{black}
% Secondly, reduction models via  balanced truncation can be extended to the cases where $g_i(s)$ has more complicated turbine dynamics, such as 2nd order turbine-governor model, or even higher-order models. 

\ifthenelse{\boolean{archive}}{
\subsection{Proof of the Theorem \ref{thm_unifm_conv}}
To proof the theorem, we need to present two lemmas first.
\begin{lem}\label{lem_bd_norm_mat_inv}
    Let A,B be matrices of order $n$. For \textbf{increasingly} ordered singular values $\sigma_i(A),\sigma_i(B)$, if $\sigma_1(A)\geq \sigma_n(B)$, then the following inequality holds:
    \ben
        \|(A+B)^{-1}\| \leq \frac{1}{\sigma_1(A)-\sigma_n(B)}=\frac{1}{\sigma_1(A)-\|B\|}
    \een
\end{lem}
\begin{proof}
    By~\cite[3.3.16]{Horn:2012:MA:2422911}, we have:
    \ben
        \sigma_{1}(A)\leq \sigma_{1}(A+B)+\sigma_n(-B)\,.
    \een
    Then as long as $\sigma_1(A)\geq \sigma_n(B)$, the following holds
    \ben
        \frac{1}{\sigma_1(A+B)}\leq \frac{1}{\sigma_1(A)-\sigma_n(B)}\,,
    \een
    and notice that the left-hand side is exactly $\|(A+B)^{-1}\|$.
\end{proof}

\begin{lem}\label{lem_reg_norm_bd}
    Let $\hat{g}(s),T(s)$ be defined in \eqref{eq_aggr_dym} and \eqref{eq_T_explicit}. Define $\bar{g}(s):=n\hat{g}(s)$. Suppose for $s_0\in\compl$, we have $|\bar{g}(s_0)|\leq M_1$ and $\max_{1\leq i\leq n}|g_i^{-1}(s_0)|\leq M_2$ for some $M_1,M_2>0$. Then for large enough $\lambda_2(L)$, the following inequality holds:
    %\begin{align}
    %    &\;\lV T(s_0)-\frac{1}{n}\bar{g}(s_0)\one\one^T\rV\nonumber\\
    %    &\;\quad\quad\leq 
    %    \frac{M_1^2M_2^2+2M_1M_2+\frac{M_1M_2^2}{|\frac{\lambda_2(L)}{s_0}|-M_2}}{|\frac{\lambda_2(L)}{s_0}|-M_2-M_1M_2^2}+\frac{1}{|\frac{\lambda_2(L)}{s_0}|-M_2}\label{eq_T_norm_bd}
    %\end{align}
    \begin{align}
        &\;\lV T(s_0)-\frac{1}{n}\bar{g}(s_0)\one\one^T\rV\nonumber\\
        \leq &\; \frac{M_1^2M_2^2+2M_1M_2+\frac{M_1M_2^2}{|\lambda_2(L)/s_0|-M_2}}{|\lambda_2(L)/s_0|-M_2-M_1M_2^2}+\frac{1}{|\lambda_2(L)/s_0|-M_2}\,.\label{eq_T_norm_bd}
    \end{align}
\end{lem}
\begin{proof}
    Since $L$ is symmetric Laplacian matrix, the decomposition of $L$ is given by:
    \begin{equation*}\label{laplcian_decomp}
    L=V\Lambda V^T\,,
    \end{equation*} where $V=[\frac{\one_n}{\sqrt{n}},V_\perp]$, $VV^T=V^TV=I_n$, and $\Lambda=\mathrm{diag}\{\lambda_i(L)\}$ with $0=\lambda_1(L)\leq\lambda_2(L)\leq \cdots\leq \lambda_n(L)$.
    
    For the transfer matrix $T(s)$, we have:
    \begin{align*}
        T(s)&=\; (I_n+\mathrm{diag}\{g_i(s)\}L/s)^{-1}\mathrm{diag}\{g_i(s)\}\nonumber\\
        &=\; (\mathrm{diag}\{g^{-1}_i(s)\}+L/s)^{-1}\nonumber\\
        &=\; (\mathrm{diag}\{g^{-1}_i(s)\}+V(\Lambda/s) V^T)^{-1}\nonumber\\
        &=\; V(V^T\mathrm{diag}\{g^{-1}_i(s)\}V+\Lambda/s)^{-1}V^T\,.
    \end{align*}
    Let $H=V^T\dg\{g_i^{-1}(s_0)\}V+\Lambda/s_0$, then it's easy to see that:
    \begin{align}
        \lV T(s_0)-\frac{1}{n}\bar{g}(s_0)\one_n\one_n^T\rV&=\; \|T(s_0)-\bar{g}(s_0)Ve_1e_1^TV^T\|\nonumber\\
        &=\; \lV V\lp H^{-1}-\bar{g}(s_0)e_1e_1^T\rp V^T\rV\nonumber\\
        &=\; \lV H^{-1}-\bar{g}(s_0)e_1e_1^T \rV\,,\label{eq_T_H_norm_equiv}
    \end{align}
    where $e_1$ is the first column of identity matrix $I_n$.
    
    We write $H$ in block matrix form:
    \begin{align*}
        H&=\;V^T\dg\{g_i^{-1}(s_0)\}V+\Lambda/s_0\\
        &= \bmt
            \frac{\one_n^T}{\sqrt{n}}\\
            V_\perp^T
        \emt \dg\{g_i^{-1}(s_0)\} \bmt
        \frac{\one_n}{\sqrt{n}}& V_\perp\emt+\Lambda/s_0\\
        &= \bmt
        \bar{g}^{-1}(s_0)& \frac{\one_n^T}{\sqrt{n}}\dg\{g_i^{-1}(s_0)\}V_\perp\\
        V_\perp^T\dg\{g_i^{-1}(s_0)\}\frac{\one_n}{\sqrt{n}} &V_\perp^T\dg\{g_i^{-1}(s_0)\}V_\perp+\Tilde{\Lambda}/s_0
        \emt\\
        &:= \bmt
        \bar{g}^{-1}(s_0)& h^T_{12}\\
        h_{12} & H_{22}
        \emt\,,
    \end{align*}
    where $\Tilde{\Lambda}=\dg\{\lambda_2(L),\cdots,\lambda_n(L)\}$.
    
    Invert $H$ in its block form, we have:
    \ben
        H^{-1} = \bmt
        a &-ah_{12}^TH_{22}^{-1}\\
        -aH_{22}^{-1}h_{12}& H_{22}^{-1}+aH_{22}^{-1}h_{12}h_{12}^TH_{22}^{-1}
        \emt\,,
    \een
    where $a = \frac{1}{\bar{g}^{-1}(s_0)-h_{12}^TH_{22}^{-1}h_{12}}$.
    
    Notice that $\|\one_n\|=\sqrt{n}$ and $\|V_\perp\|=1$, we have
    \be
        \|h_{12}\|\leq \frac{\|\one_n\|}{\sqrt{n}}\|\dg\{g_i^{-1}(s_0)\}\|\|V_\perp\|\leq M_2\,,\label{eq_h12_norm_bd}
    \ee
    by the compatibility between vector and matrix 2-norm, along with that matrix 2-norm is sub-multiplicative.
    Additionally, by Lemma \ref{lem_bd_norm_mat_inv}, when $|\lambda_2(L)/s_0|>M_2$, the following holds:
    \begin{align}
        \|H_{22}^{-1}\|&\leq\; \frac{1}{\sigma_1(\Tilde{\Lambda})-\|V_\perp^T\dg\{g_i^{-1}(s_0)\}V_\perp\|}\nonumber\\
        &\leq\; \frac{1}{|\lambda_2(L)/s_0|-M_2}\,.\label{eq_H22_norm_bd}
    \end{align}
    Lastly, when $|\lambda_2(L)/s_0|>M_2+M_2^2M_1$, by \eqref{eq_h12_norm_bd}\eqref{eq_H22_norm_bd}, we have:
    \begin{align}
        |a|&\leq\; \frac{1}{|\bar{g}^{-1}(s_0)|-\|h_{12}\|^2\|\|H_{22}^{-1}\|}\nonumber\\
        &\leq \; \frac{(|\lambda_2(L)/s_0|-M_2)M_1}{|\lambda_2(L)/s_0|-M_2-M_1M_2^2}\,.\label{eq_a_norm_bd}
    \end{align}
    
    Now we bound the norm of $H^{-1}-\bar{g}(s_0)e_1e_1^T$ by the sum of norms of all its blocks:
    \begin{align}
        &\;\|H^{-1}-\bar{g}(s_0)e_1e_1^T\|\nonumber\\
        =&\; \lV \bmt
        a\bar{g}(s_0)h_{12}^TH_{22}^{-1}h_{12} &-ah_{12}^TH_{22}^{-1}\\
        -aH_{22}^{-1}h_{12}& H_{22}^{-1}+aH_{22}^{-1}h_{12}h_{12}^TH_{22}^{-1}
        \emt\rV\nonumber\\
        \leq &\; |a|\|H_{22}^{-1}\|(|\bar{g}(s_0)|\|h_{12}\|^2+2\|h_{12}\|+\|h_{12}\|^2\|H_{22}^{-1}\|)\nonumber\\
        &\;\quad\quad +\|H_{22}^{-1}\|\,.\label{eq_Hinv_norm_bd1}
    \end{align}
    By \eqref{eq_h12_norm_bd}\eqref{eq_H22_norm_bd}\eqref{eq_a_norm_bd}, we have the following:
    \begin{align}
        &\;\|H^{-1}-\bar{g}(s_0)e_1e_1^T\|\nonumber\\
        \leq &\; \frac{M_1^2M_2^2+2M_1M_2+\frac{M_1M_2^2}{|\lambda_2(L)/s_0|-M_2}}{|\lambda_2(L)/s_0|-M_2-M_1M_2^2}+\frac{1}{|\lambda_2(L)/s_0|-M_2}\,.\label{eq_Hinv_norm_bd2}
    \end{align}
    This bound holds as long as $|\lambda_2(L)/s_0|>M_2+M_2^2M_1$, and combining \eqref{eq_T_H_norm_equiv}\eqref{eq_Hinv_norm_bd2} gives the desired inequality.
\end{proof}
Now we can proof theorem \ref{thm_unifm_conv}, we recite the theorem before the proof:

\begin{usethmcounterof}{thm_unifm_conv}
    Given the assumptions in Section \ref{ssec:coherence}, the following holds for any $\eta_0>0$:
    \ben
        \lim_{\lambda_2(L)\ra +\infty}\sup_{\eta\in[-\eta_0,\eta_0]}\lV T(j\eta)-\hat{g}(j\eta)\one\one^T\rV=0\,,
    \een
    where $j=\sqrt{-1}$ and $\one\in\mathbb{R}^n$ is the vector of all ones.
\end{usethmcounterof}
\begin{proof}
    $\bar{g}(s)$ is stable because $\hat{g}(s)$ is stable, then $\bar{g}(s)$ is continuous on compact set $[-j\eta_0,j\eta_0]$. Then by~\cite[Theorem 4.15]{Rudin1964} there exists $M_1>0$, such that $\forall s\in[-j\eta_0,j\eta_0]$, we have $|\bar{g}(s)|\leq M_1$. Similarly, because all $g_i(s)$ are minimum-phase, all $g_i^{-1}(s)$ are stable hence continuous on $[-j\eta_0,j\eta_0]$. Again there exists $M_2>0$, such that $\forall s\in[-j\eta_0,j\eta_0]$, we have $\max_{1\leq i\leq n}|g_i^{-1}(s)|\leq M_2$.
    
    Now we know that $\forall s\in [-j\eta_0,j\eta_0] $, we have $|\bar{g}(s)|\leq M_1,\max_{1\leq i\leq n}|g_i^{-1}(s)|\leq M_2$, i.e. the condition for Lemma~\ref{lem_reg_norm_bd} is satisfied for a common choice of $M_1,M_2>0$.
    
    By Lemma~\ref{lem_reg_norm_bd}, $\forall s\in[-j\eta_0,j\eta_0]$, we have:
    \begin{align*}
        &\;\lV T(s)-\hat{g}(s)\one\one^T\rV\nonumber\\
        \leq &\; \frac{M_1^2M_2^2+2M_1M_2+\frac{M_1M_2^2}{|\lambda_2(L)/s|-M_2}}{|\lambda_2(L)/s|-M_2-M_1M_2^2}+\frac{1}{|\lambda_2(L)/s|-M_2}.
    \end{align*}
    Taking $\sup_{s\in [-j\eta_0,j\eta_0]}$ on both sides gives:
    \begin{align*}
        &\;\sup_{s\in [-j\eta_0,j\eta_0]}\lV T(s)-\hat{g}(s)\one\one^T\rV\nonumber\\
        \leq &\; \frac{M_1^2M_2^2+2M_1M_2+\frac{M_1M_2^2}{|\lambda_2(L)|/\eta_0-M_2}}{|\lambda_2(L)|/\eta_0-M_2-M_1M_2^2}+\frac{1}{|\lambda_2(L)|/\eta_0-M_2}.
    \end{align*}
    Lastly, take $\lambda_2(L)\ra +\infty$ on both sides, the right-hand side gives $0$ in the limit, which finishes the proof.
\end{proof}
\subsection{Frequency Weighted balanced Truncation}\label{app_freq_weight_bt}

Given a minimum realization of frequency weight $W(s)$ to be $(A_W,B_W,C_W,D_W)$, the procedures of frequency weighted balanced truncation for a minimum, strictly proper and stable linear system $(A,B,C)$ with order $n$ are given as follow:
\begin{enumerate}
    \item The extended system\footnote{When $W(s)=1$, the extended system is exactly the same as original $(A,B,C)$, then the procedures give unweighted standard balanced truncation.} is given by:
    \ben
        \left[
        \begin{array}{cc|c}
             A&\mathbb{0}&B  \\
             B_WC&A_W&\mathbb{0}\\\hline
             D_WC&C_W&\mathbb{0}
        \end{array}
        \right]:=\left[\arraycolsep=1.4pt\def\arraystretch{1.3}
        \begin{array}{c|c}
             \bar{A}&\bar{B}  \\\hline
             \bar{C}&\mathbb{0}
        \end{array}
        \right].
    \een
    \item Compute the frequency weighted controllability and observability gramians $X_c,Y_o$ from the gramians $\bar{X}_c,\bar{Y}_o$ of extended system:
    \ben
        \bar{X}_c=\int_{0}^{\infty}e^{\bar{A}t}\bar{B}\bar{B}^Te^{\bar{A}^Tt}dt,\ \bar{Y}_o = \int_{0}^{\infty} e^{\bar{A}^Tt}\bar{C}^T\bar{C}e^{\bar{A}t}dt
    \een
    \ben
        X_c = \begin{bmatrix}
            I_{n}& \mathbb{0}
        \end{bmatrix}\bar{X_c}\begin{bmatrix}
            I_{n}\\ \mathbb{0}
        \end{bmatrix},\ Y_c = \begin{bmatrix}
            I_{n}& \mathbb{0}
        \end{bmatrix}\bar{Y_c}\begin{bmatrix}
            I_{n}\\ \mathbb{0}
        \end{bmatrix}\,.
    \een
    \item Perform the singular value decomposition of $X_c^{\frac{1}{2}}Y_oX_c^{\frac{1}{2}}$:
    \ben
        X_c^{\frac{1}{2}}Y_oX_c^{\frac{1}{2}} = U\Sigma U^*\,.
    \een
    where $U$ is unitary and $\Sigma$ is diagonal, positive definite with its diagonal terms in decreasing order. Then compute the change of coordinates $T$ given by:
    \ben
        T^{-1} = X_c^{\frac{1}{2}}U\Sigma^{-1}\,.
    \een
    \item Apply change of coordinates $T$ on $(A,B,C)$ to get its balanced realization $(TAT^{-1},TB,CT^{-1})$. Then the $k$-th order $(1\leq k\leq n)$ reduction model $(A_k,B_k,C_k)$ is given by truncating $(TAT^{-1},TB,CT^{-1})$ as the following:
    \begin{align*}
        &\;A_k = \begin{bmatrix}
            I_{k}& \mathbb{0}
        \end{bmatrix}TAT^{-1}\begin{bmatrix}
            I_{k}\\ \mathbb{0}
        \end{bmatrix}\\
        &\;B_k = \begin{bmatrix}
            I_{k}& \mathbb{0}
        \end{bmatrix}TB\\
        &\; C_k =CT^{-1}\begin{bmatrix}
            I_{k}\\ \mathbb{0}
        \end{bmatrix}\,.
    \end{align*}
    
\end{enumerate}
\begin{rem}
    Balanced truncation only applies to systems in state space. For a transfer function, one should apply balanced truncation to its minimum realization, then obtain reduced order transfer function from the state-space reduction model. 
\end{rem}
}{}

\bibliographystyle{IEEEtran}
\bibliography{ref.bib}

\begin{thebibliography}{10}
\providecommand{\url}[1]{#1}
\csname url@samestyle\endcsname
\providecommand{\newblock}{\relax}
\providecommand{\bibinfo}[2]{#2}
\providecommand{\BIBentrySTDinterwordspacing}{\spaceskip=0pt\relax}
\providecommand{\BIBentryALTinterwordstretchfactor}{4}
\providecommand{\BIBentryALTinterwordspacing}{\spaceskip=\fontdimen2\font plus
\BIBentryALTinterwordstretchfactor\fontdimen3\font minus
  \fontdimen4\font\relax}
\providecommand{\BIBforeignlanguage}[2]{{%
\expandafter\ifx\csname l@#1\endcsname\relax
\typeout{** WARNING: IEEEtran.bst: No hyphenation pattern has been}%
\typeout{** loaded for the language `#1'. Using the pattern for}%
\typeout{** the default language instead.}%
\else
\language=\csname l@#1\endcsname
\fi
#2}}
\providecommand{\BIBdecl}{\relax}
\BIBdecl

\bibitem{Chow2013}
J.~H. Chow, \emph{Power system coherency and model reduction}.\hskip 1em plus
  0.5em minus 0.4em\relax Springer, 2013.

\bibitem{podmore1978identification}
R.~Podmore, ``Identification of coherent generators for dynamic equivalents,''
  \emph{{IEEE} Trans. Power App. Syst.}, no.~4, pp. 1344--1354, 1978.

\bibitem{Souza1992efficient}
E.~P. de~Souza and A.~L. da~Silva, ``An efficient methodology for
  coherency-based dynamic equivalents,'' in \emph{IEE Proceedings C
  (Generation, Transmission and Distribution)}, vol. 139, no.~5.\hskip 1em plus
  0.5em minus 0.4em\relax IET, 1992, pp. 371--382.

\bibitem{chow1982time}
J.~H. Chow, G.~Peponides, P.~Kokotovic, B.~Avramovic, and J.~Winkelman,
  \emph{Time-scale modeling of dynamic networks with applications to power
  systems}.\hskip 1em plus 0.5em minus 0.4em\relax Springer, 1982, vol.~46.

\bibitem{Winkelman1981}
J.~R. {Winkelman}, J.~H. {Chow}, B.~C. {Bowler}, B.~{Avramovic}, and P.~V.
  {Kokotovic}, ``An analysis of interarea dynamics of multi-machine systems,''
  \emph{{IEEE} Trans. Power App. Syst.}, vol. PAS-100, no.~2, pp. 754--763, Feb
  1981.

\bibitem{Nath1985}
R.~{Nath}, S.~S. {Lamba}, and K.~s.~P.~{Rao}, ``Coherency based system
  decomposition into study and external areas using weak coupling,''
  \emph{{IEEE} Trans. Power App. Syst.}, vol. PAS-104, no.~6, pp. 1443--1449,
  June 1985.

\bibitem{Podmore2013}
R.~Podmore, \emph{Coherency in Power Systems}.\hskip 1em plus 0.5em minus
  0.4em\relax New York, NY: Springer New York, 2013, pp. 15--38.

\bibitem{Anderson1990}
P.~M. Anderson and M.~Mirheydar, ``A low-order system frequency response
  model,'' \emph{{IEEE} Trans. Power Syst.}, vol.~5, no.~3, pp. 720--729, 1990.

\bibitem{Germond1978}
A.~J. {Germond} and R.~{Podmore}, ``Dynamic aggregation of generating unit
  models,'' \emph{{IEEE} Trans. Power App. Syst.}, vol. PAS-97, no.~4, pp.
  1060--1069, July 1978.

\bibitem{Guggilam2018}
S.~S. {Guggilam}, C.~{Zhao}, E.~{Dall’Anese}, Y.~C. {Chen}, and S.~V.
  {Dhople}, ``Optimizing {DER} participation in inertial and primary-frequency
  response,'' \emph{{IEEE} Trans. Power Syst.}, vol.~33, no.~5, pp. 5194--5205,
  Sep. 2018.

\bibitem{Apostolopoulou2016}
D.~{Apostolopoulou}, P.~W. {Sauer}, and A.~D. {Dom\' inguez-Garc\'ia},
  ``Balancing authority area model and its application to the design of
  adaptive {AGC} systems,'' \emph{{IEEE} Trans. Power Syst.}, vol.~31, no.~5,
  pp. 3756--3764, Sep. 2016.

\bibitem{ourari2006dynamic}
M.~L. Ourari, L.-A. Dessaint, and V.-Q. Do, ``Dynamic equivalent modeling of
  large power systems using structure preservation technique,'' \emph{{IEEE}
  Trans. Power Syst.}, vol.~21, no.~3, pp. 1284--1295, 2006.

\bibitem{Paganini2019tac}
F.~{Paganini} and E.~{Mallada}, ``Global analysis of synchronization
  performance for power systems: Bridging the theory-practice gap,''
  \emph{{IEEE} Trans. Automat. Contr.}, vol.~65, no.~7, pp. 3007--3022, 2020.

\bibitem{jiang2020grid}
Y.~Jiang, A.~Bernstein, P.~Vorobev, and E.~Mallada, ``Grid-forming frequency
  shaping control,'' \emph{arXiv preprint arXiv:2009.06707}, 2020.

\bibitem{schiffer2014conditions}
J.~Schiffer, R.~Ortega, A.~Astolfi, J.~Raisch, and T.~Sezi, ``Conditions for
  stability of droop-controlled inverter-based microgrids,'' \emph{Automatica},
  vol.~50, no.~10, pp. 2457--2469, 2014.

\bibitem{tegling2015performance}
E.~Tegling, D.~F. Gayme, and H.~Sandberg, ``Performance metrics for
  droop-controlled microgrids with variable voltage dynamics,'' in \emph{IEEE
  54th Conf. on Decision and Control}.\hskip 1em plus 0.5em minus 0.4em\relax
  IEEE, 2015, pp. 7502--7509.

\bibitem{min2019cdc}
H.~{Min} and E.~{Mallada}, ``Dynamics concentration of large-scale
  tightly-connected networks,'' in \emph{IEEE 58th Conf. on Decision and
  Control}, 2019, pp. 758--763.

\bibitem{min2019allerton}
H.~{Min}, F.~{Paganini}, and E.~{Mallada}, ``Accurate reduced order models for
  coherent synchronous generators,'' in \emph{57th Annual Allerton Conf. on
  Communication, Control, and Computing}, 2019, pp. 316--317.

\bibitem{zhao2013power}
C.~Zhao, U.~Topcu, N.~Li, and S.~Low, ``Power system dynamics as primal-dual
  algorithm for optimal load control,'' \emph{arXiv preprint arXiv:1305.0585},
  2013.

\bibitem{khalil2002nonlinear}
H.~Khalil, \emph{Nonlinear Systems}.\hskip 1em plus 0.5em minus 0.4em\relax
  Prentice Hall, 2002.

\bibitem{iceland}
\BIBentryALTinterwordspacing
U.~of~Edinburgh. Power systems test case archive. [Online]. Available:
  \url{https://www.maths.ed.ac.uk/optenergy/NetworkData/icelandDyn/}
\BIBentrySTDinterwordspacing

\bibitem{Zhou:1996:ROC:225507}
K.~Zhou, J.~C. Doyle, and K.~Glover, \emph{Robust and Optimal Control}.\hskip
  1em plus 0.5em minus 0.4em\relax Upper Saddle River, NJ, USA: Prentice-Hall,
  Inc., 1996.

\bibitem{SANGWOOKIM1995183}
S.~W. Kim, B.~D. Anderson, and A.~G. Madievski, ``Error bound for transfer
  function order reduction using freqeuncy weighted balanced truncation,''
  \emph{Systems \& Control Letters}, vol.~24, no.~3, pp. 183 -- 192, 1995.

\bibitem{jiang2017performance}
Y.~Jiang, R.~Pates, and E.~Mallada, ``Performance tradeoffs of dynamically
  controlled grid-connected inverters in low inertia power systems,'' in
  \emph{IEEE 56th Conf. on Decision and Control}.\hskip 1em plus 0.5em minus
  0.4em\relax IEEE, 2017, pp. 5098--5105.

\bibitem{Horn:2012:MA:2422911}
R.~A. Horn and C.~R. Johnson, \emph{Matrix Analysis}, 2nd~ed.\hskip 1em plus
  0.5em minus 0.4em\relax New York, NY, USA: Cambridge University Press, 2012.

\bibitem{Rudin1964}
W.~Rudin \emph{et~al.}, \emph{Principles of mathematical analysis}.\hskip 1em
  plus 0.5em minus 0.4em\relax McGraw-hill New York, 1964, vol.~3.

\end{thebibliography}

\balance
\end{document}